\documentclass[jmp]{revtex4-1}
\usepackage{amssymb}
\usepackage{amsmath}
\usepackage{amsfonts}
\usepackage{amsthm}
\usepackage{graphicx}
\usepackage{ulem}
\usepackage{bbm}
\usepackage{mathrsfs}

\newtheorem*{rpth}{Regular Perturbation Theorem}

\newtheorem*{lemma}{Lemma}

\def\ee{{\varepsilon}}
\def\ex{{ \mathbf{e}_{1} }}
\def\ey{{ \mathbf{e}_{2} }}
\def\ez{{ \mathbf{e}_{3} }}
\def\eo{{ E_{\omega}}}
\def\eot{{ \widetilde{E}_{\omega} }}
\def\eos{{\widetilde{E}_{\omega}^{s}}}
\def\rr{{ \mathbf{r}}}
\def\ro{{ \mathbf{r}_{0}}}
\def\kk{{ \mathbf{k} }}
\def\kx{{ k_{x} }}

\def\kmm{{k_{x,m}^{2}}}
\def\kz{{ k_{z} }}

\def\kc{{ k_{c}}}
\def\hb{{ h_{b} }}

\def\es{{ \zeta }}
\def\eh{{ \zeta_{n}}}
\def\ehf{{ \zeta_{n}}}
\def\sih{{E_{n}}}

\def\ib{{\mathrm{I}}}

\def\hv{{ \widehat{\mathrm{H}} }}
\def\pv{{ \widehat{\mathrm{P}}}}
\def\ckx{{ \mathbbm{C}_{k_{x}}}}
\def\se{{ \mathcal{S} }}
\def\re{{ \text{Re} }}
\def\im{{ \text{Im} }}

\begin{document}

\title{Electromagnetic Siegert states for periodic dielectric structures}

\author{ R\'{e}my F. Ndangali and Sergei V. Shabanov \\  Department of Mathematics, University of Florida, Gainesville, FL 32611, USA}

\begin{abstract}
The formalism of Siegert states to describe the resonant
scattering in quantum theory is extended to the resonant
scattering of electromagnetic waves on periodic
dielectric arrays.
The excitation of
electromagnetic Siegert
states by an incident wave packet and their decay is
studied.
The formalism
is applied to develop a theory of coupled electromagnetic
resonances arising in the electromagnetic scattering problem
for two such arrays separated by a distance $2h$ (or, generally,
when the physical properties of the scattering array depend
on a real coupling parameter $h$).
Analytic properties
of Siegert states as functions of the coupling parameter $h$
are established by the Regular Perturbation Theorem which is
an extension the Kato-Rellich theorem to the present case.
By means of this theorem, it is proved that if the scattering
structure admits a bound state in the radiation continuum
at a certain value of the coupling parameter $h$, then there
always exist regions within the structure in which the near field
can be amplified as much as desired by adjusting
the value of $h$. This establishes a rather general mechanism
to control and amplify optical nonlinear effects in periodically
structured planar structures possessing a nonlinear dielectric
susceptibility.
\end{abstract}

\maketitle

\section{Introduction}
\label{sec:1}

To outline the scope of the problem studied, it is
helpful first to glean over the basic concepts of quantum
resonant scattering theory.
Consider the scattering problem $\hv \Psi ={\cal E}\Psi$
 for a quantum mechanical
system described by the Hamiltonian
$\hv = -\frac 12 \Delta + V({\bf r})$
where $\Delta $ is the Laplace operator in space,
${\cal E}$ is the spectral parameter (the energy), and
the scattering potential $V({\bf r})$ is assumed
to be of a finite range, i.e., $V({\bf r})=0$
if $r=|{\bf r}|>r_0$ for some $r_0$.
To avoid excessive and unnecessary technicalities,
the potential is assumed
to be spherically symmetric, $V({\bf r})=V(r)$.
Suppose further that only a spherically symmetric
scattered wave is of interest (the $s-$wave).
Then the scattering theory \cite{b4,res} requires that in the asymptotic
region $r>r_0$
an eigenfunction
of the Hamiltonian $\hv$ corresponding to an eigenvalue
${\cal E}$ is the superposition of an incident
wave $\Psi_i\sim e^{i{\bf k}\cdot{\bf r}}$ and an
outgoing wave $\Psi_s$:
\begin{equation*}
\label{eq:i1}
\Psi=\Psi_i +\Psi_s,\quad  \Psi_s= S\, \frac{e^{ikr}}{r}\,,\quad
k=|{\bf k}|=\sqrt{2{\cal E}}
\end{equation*}
where $S$ is called the scattering amplitude . The system is said
to have a scattering resonance if $S$ has a pole in the complex
energy plane:
$$
S \sim \frac{\Gamma_n}{{\cal E}-{\cal E}_n+i\Gamma_n}
$$
where ${\cal E}_n$ is the resonance position (or the resonant
energy) and $\Gamma_n$ is the resonance width. In this case
the scattering cross section as a function of ${\cal E}$
exhibits a Lorentzian profile
\begin{equation*}
\sigma({\cal E})\sim |S|^2\sim
\frac{\Gamma_{n}^{2}}{({\cal E}-{\cal E}_{n})^2+\Gamma_{n}^{2}}
\end{equation*}
In 1939, A. Siegert showed \cite{Siegert}
that positions and widths of resonances
for a given quantum system can be determined by solving
the eigenvalue problem
under the outgoing wave boundary condition which
in the simplest case of a spherically symmetric scattered wave
(the $s-$wave) reads
$$
\hv \Psi_n = {\cal E}\Psi_n\,,\quad
(\partial_r(r\Psi_n) - ikr \Psi_n)\Bigr|_{r\geq r_0} = 0
$$
Under this boundary condition the Hamiltonian is no longer
hermitian, hence, the eigenvalues are generally complex,
${\cal E} = {\cal E}_n -i\Gamma_n$. The eigenstates $
\Psi_n$ associated
with  eigenvalues ${\cal E}_n -i\Gamma_n$
are now called Siegert states.
The scattered wave is then proved to have the form
$$
\Psi_s = \sum_n \frac{a_n}{{\cal E}-{\cal E}_n+i\Gamma_n}\,
\Psi_n + \Psi_a
$$
where the complex amplitudes $a_n=a_n[\Psi_i]$ are homogeneous linear functionals
of the incident wave and $\Psi_a$ is the so called
background or potential scattering which is analytic in ${\cal E}$ (see e.g.,\cite{mlthm,ostrovsky}).
If $\Psi_i$ is a wave packet with a narrow energy
distribution centered at ${\cal E}={\cal E}_c$ (e.g., a
narrow Gaussian
wave packet), then the amplitude $a_n$ is significant, provided
${\cal E}_c \approx {\cal E}_n$. In other words, a Siegert
state contributes to the scattering wave if the incident wave
has a resonant energy ${\cal E}_n$ determined by the real
part of the corresponding Siegert eigenvalue.

Note that the scattered modes always have positive energy
${\cal E}=k^2/2>0$. The range of ${\cal E}$ corresponding
to scattered modes is called {\it the radiation continuum}.
Suppose that $V=V(r)$, for simplicity. It is easy to see
that the Siegert states contain bound states of the system
${\cal E}_n <0$ and $\Gamma_n=0$ which decay exponentially
in the asymptotic region (as $k=i\sqrt{-2{\cal E}_n}$) and, hence,
are square integrable eigenstates of the Hamiltonian.
The bound states have a discrete spectrum which
lies {\it below} the radiation continuum (${\cal E}_n<0$).
If the asymptotic behavior of the potential $V$ is relaxed
by demanding that $V(r)\rightarrow 0$ as $r\rightarrow \infty$,
then there may exist bound states (i.e., square integrable solutions
of the stationary Schr\"odinger equation) with positive energies
${\cal E}_n>0$, $\Gamma_n=0$.
They are called {\it bound states in the radiation
continuum} or {\it resonances with the vanishing width}
(to emphasize the fact ${\cal E}_n>0$).
Their existence was first predicted
by von Neumann and Wigner in 1929 \cite{nw1929}.
It can further be shown that
the amplitudes $a_n$ vanish if $\Gamma_n=0$. Thus, no bound
state (either above or below the radiation continuum) of the system can be excited by an incident wave \cite{b4}.

The physical significance of Siegert states with $\Gamma_n>0$ can be
understood through the initial value problem for
the time dependent Schr\"odinger equation $i\partial_t \Psi=
\hv \Psi$ in which a wave packet $\Psi_i$, initially positioned
in the asymptotic region $r>r_0$ (i.e., at the initial time $t=0$
the support of $\Psi_i$ lies in the asymptotic region),
propagates into the scattering region $r<r_0$ and passes through it
giving rise to scattered waves. In the time-dependent picture,
the amplitude of each Siegert state (with a sufficiently
small $\Gamma_n$) that has been excited by the incident
wave packet is shown to decay exponentially \cite{b4}:
$$
\Psi_s(t) =\sum_n \Omega_n(t) \Psi_n +\Psi_a(t)\,,\quad
\Omega_n(t) \sim e^{-i{\cal E}_n t }e^{-\Gamma_n t}
$$
If the incident wave packet
passes through the scattering region faster than
the decay time $\tau_n =1/\Gamma_n$ of a Siegert state,
then the outgoing wave can be observed, and it resembles to
a stationary state with the energy ${\cal E}_n$. The more narrow the scattering resonance is, the longer
lives the corresponding Siegert state. So, Siegert states
with $\Gamma_n\ll 1$
may be interpreted as quasi-stationary states of the system
that can be excited by an incident wave and live long after
the scattering process is over.

Resonant scattering phenomena are also quite common in
electromagnetic theory. In the past decade much attention has been
devoted to experimental and theoretical studies of
reflection and transmission properties of
periodically perforated films and similar periodic planar structures
in infrared and visible light. The transmission and reflection
coefficients of such structures exhibit typical resonant peaks
at a wavelength about the structure period \cite{1998}
(for a review see \cite{abajo,svs}). Numerical studies
of the scattering problem for an electromagnetic wave packet
impinging a periodic structure show that a portion
of the electromagnetic energy of the wave packet is trapped
by the structure and remains in it long after the wave packet
has passed the structure \cite{bs1,bs2,bs3,svs}. Furthermore, each "trapped" electromagnetic mode decays slowly by emitting a monochromatic
radiation whose wave length is close to the resonant wavelength.
Broad and narrow resonances correspond to short-lived and
long-lived "trapped" modes, respectively.
It has been shown that
"trapped" modes may occur either due to the structure geometry
or the dispersive properties of the structure
material \cite{bs2,svs}.

If two planar periodic
structures are separated by a distance $2h$, then the resonance
positions and their widths depend on $h$ so that under some
conditions, there are critical values $h=h_b$ at which the widths
of some resonances approach
zero as $h\rightarrow h_b$ \cite{b5,svs,b6}.
Studies of specific examples show that when $h=h_b$
the structure
has a stable wave guiding mode localized near the structure
which cannot be excited by an incident radiation despite the fact
that the spectral parameters of this mode lie in the spectral
range of the scattered modes (the radiation continuum).
In other words, such wave guiding modes are a new type of
localized solutions of Maxwell's equations which are analogous
to quantum mechanical bounds states in the radiation continuum.

It has been observed that in the spectral range near these localized solutions, there is a local amplification
of near-fields (as compared to the amplitude of the incident
radiation) for periodic scattering structures. The amplification
effect becomes stronger if the system exhibits
more narrow resonances \cite{bs3,b5,svs,b6}.

The similarities between resonant quantum and electromagnetic
scattering phenomena are evident. It is therefore natural
to develop a similar mathematical formalism of Siegert states
to study the resonant electromagnetic scattering for periodic
structures. Although a qualitative analogy between electromagnetic
"trapped" modes and quantum Siegert states has been drawn,
no rigorous quantitative studies have been carried out.
It appears that there are two substantial differences between
the theories of quantum mechanics and electromagnetism which
prevent one from a trivial extension of quantum mechanical Siegert
states to electromagnetism.

First, the electric field ${\bf E}$
propagating in a non-homogeneous medium with a dielectric susceptibility
$\varepsilon$
satisfies the wave equation
$\frac{\varepsilon}{c^2} \partial_t^2 {\bf E}= \Delta {\bf E}$.
Consequently, the stationary analog of the Schr\"{o}dinger
equation, which is obtained when ${\bf E}({\bf r},t)=
e^{-i\omega t} {\bf E}_\omega({\bf r})$, is:
$$
\Bigl[-\Delta +k^2(1-\varepsilon({\bf r}))\Bigr]{\bf E}_\omega
=k^2 {\bf E}_\omega,\, k^2=\frac{\omega^2}{c^2}
$$
This equation can still be viewed as an eigenvalue problem,
but in contrast to the stationary
Schr\"odinger equation, the "potential"
$V(\rr)=k^2(1-\varepsilon({\bf r}))$ depends on
the spectral parameter $k^2$. Furthermore, the dielectric
susceptibility could also be a complex-valued function of $k$ if
the medium is dispersive, i.e.,
$\varepsilon=\varepsilon(k,{\bf r})$.

Second, for a periodic scattering
structure, the amplitude ${\bf E}_\omega$ must obey Bloch's
periodicity
condition which is not compatible with the asymptotic boundary
condition for quantum mechanical Siegert states. The stated differences require a modification of the very definition
of Siegert states in electromagnetic theory.
This problem is investigated in the present study.

In Section~\ref{sec:2}, electromagnetic Siegert states are defined
for periodic scattering structures by means of the theory
of compact integral operators when the polarization of the
electromagnetic wave is preserved, i.e., when the vector wave
equation is reduced to the scalar one. A relation between
the Siegert
states and the electromagnetic resonant scattering of waves is established.
In Section~\ref{sec:6}, the exponential decay
of electromagnetic Siegert states excited by an incident wave packet is investigated. This
allows one to identify the constructed Siegert states with
"trapped" quasi-stationary electromagnetic modes observed
in aforementioned numerical studies of periodic structures.
In quantum mechanics, resonances of a system that consists
of two coupled resonant systems depend on a coupling parameter
$h$, and this dependence is determined by the theory of
coupled resonances. A similar problem for electromagnetic
resonances is studied in Section~\ref{sec:3}. In the said section, the central result is
an extension of the Kato-Rellich theorem~\cite{b3} on analytic perturbation of resonance poles in quantum mechanics. In a similar fashion, the theorem established states analytic properties of
electromagnetic Siegert states as functions of the coupling
parameter $h$ (the Regular Perturbation Theorem).
By means of this theorem, it is then proved that
if the scattering system admits a bound state in the radiation
continuum (a resonance with the vanishing width) at some
value of the coupling parameter $h=h_b$, then the near-field
can be amplified as much as desired by a suitable adjustment
of the value of $h$ near $h_b$. This result provides a general mechanism
to enhance optical nonlinear effects in periodic structures.
Its efficiency has been recently demonstrated with a specific
example to generate second harmonics by a double array
of dielectric cylinders \cite{shg} (as depicted in Fig.~\ref{fig:1}(b)).

\section{Electromagnetic Siegert states}
\label{sec:2}

Consider a scattering problem for an electromagnetic
plane wave impinging a structure made of a non-dispersive
dielectric. The structure is described by
a dielectric function $\varepsilon({\bf r})$
that defines the value
of the dielectric constant at every point ${\bf r}$
occupied by
the structure and it equals the dielectric constant
of the surrounding medium otherwise. Without loss
of generality, the surrounding medium is assumed to
be the vacuum, i.e., $\varepsilon = 1$.
The material is said to be non-dispersive
if its dielectric constant does not depend on the frequency
of the incident wave. The structure is assumed to have a
translational symmetry along
a particular direction. In this case the dielectric function
is independent of one of the spatial coordinates, say,
the $y$ coordinate. If the incident wave is polarized
along the $y-$axis (the electric field is parallel to
the $y-$axis), then the scattering problem can be formulated
as the scalar Maxwell's equation:
\begin{equation}
\frac{\ee}{c^2}\partial_{t}^{2}E=\Delta E,\quad
\Delta=\partial_{x}^{2}+\partial_{z}^{2}
\label{eq:1}
\end{equation}
where $c$ is the speed of light in the surrounding medium
(the vacuum), and $E$ is the electric field.
For a planar structure, the function $\ee$ differs from $1$
only within a strip $(x,z)\in (-\infty,\infty)\times(-a,a)$.
A planar structure is periodic if
$\ee(x+D_{g},z)=\ee(x,z)$ where $D_{g}$ is the period.
In what follows, the units of length are chosen so that
$D_{g}=1$. It is further assumed that the structure's dielectric
constant exceeds that of the surrounding medium, i.e.,
$\ee({\bf r})\geq 1$ for all ${\bf r}$, and $\ee$ is bounded.
Two examples of such structures are shown in
Fig.~\ref{fig:1} (Panels (a) and (b)) depicting periodic
arrays of infinite
dielectric cylinders, all parallel to the $y-$axis.
The function
$\ee({\bf r})$ is piecewise continuous. For general
periodic structures, $\ee$ is invariant under the action
of a suitable affine Coxeter group in the $xy-$plane.
In this case, vector Maxwell's equations should be used
as the polarization of the incident wave is not preserved
in the scattering process. This general problem will not be
studied here.

If the incident wave has a fixed frequency, then
the field $E$ has a harmonic time
dependence $E(\rr,t)=\eo(\rr)e^{-i\omega t}$ in
Eq. (\ref{eq:1}) where the
amplitude $\eo$ satisfies the equation:
\begin{equation}
\Delta \eo(\rr)+k^2\ee(\rr)\eo(\rr)=0,\quad
k^2=\frac{\omega^2}{c^2}
\label{eq:2}
\end{equation}
Let the $x,\, y$, and $z$ axes be oriented by the unit vectors $\ex$, $\ey$, and $\ez$, respectively.
In the scattering theory, $\eo$ is sought as a superposition of an incident wave, and the corresponding scattered
wave $E_\omega^s$:
\begin{equation}
\eo(\rr)=e^{i\kk\cdot\rr}+E_{\omega}^{s}(\rr),\quad
\kk=\kx\ex+\kz\ez
\label{eq:3}
\end{equation}
where $\kk$ is the wave vector of the incident wave.
The units of the field $E$ are chosen so that the amplitude
of the incident wave is 1.
The scattered wave
obeys the outgoing wave boundary conditions at
the spatial infinity. The periodicity of the scattering structure
requires that the amplitude $\eo$ satisfies Bloch's periodicity condition,
\begin{equation}
\eo(x+1,z)=e^{i\kx}\eo(x,z)
\label{eq:4}
\end{equation}
Under the specified boundary conditions, the amplitude
$E_\omega$ satisfies
the Lippmann-Scwinger integral equation:
\begin{equation}
\eo(\rr)=\hv[ \eo](\rr)+e^{i\kk\cdot\rr}
\label{eq:5}
\end{equation}
where $\hv$ is the integral operator defined by,
\begin{equation*}
\hv[E](\rr)=\frac{k^2}{4\pi}\int_{\mathbbm{R}^2}
(\ee(\ro)-1)G_{k}(\rr|\ro)E(\ro) d\ro
\label{eq:6}
\end{equation*}
The kernel $G_{k}(\rr|\ro)$ is the Green's function for
the Helmholtz equation, $(\Delta +k^2) G_{k}(\rr|\ro)=
-4\pi\delta(\rr-\rr_0)$,
with the outgoing boundary condition.
It is given by $G_{k}(\rr|\ro)=i\pi H_{0}(k|\rr-\ro|)$
where $H_{0}$ is the Hankel function of the first
kind of order 0.

\begin{figure}[h t]
    \centering
    \includegraphics{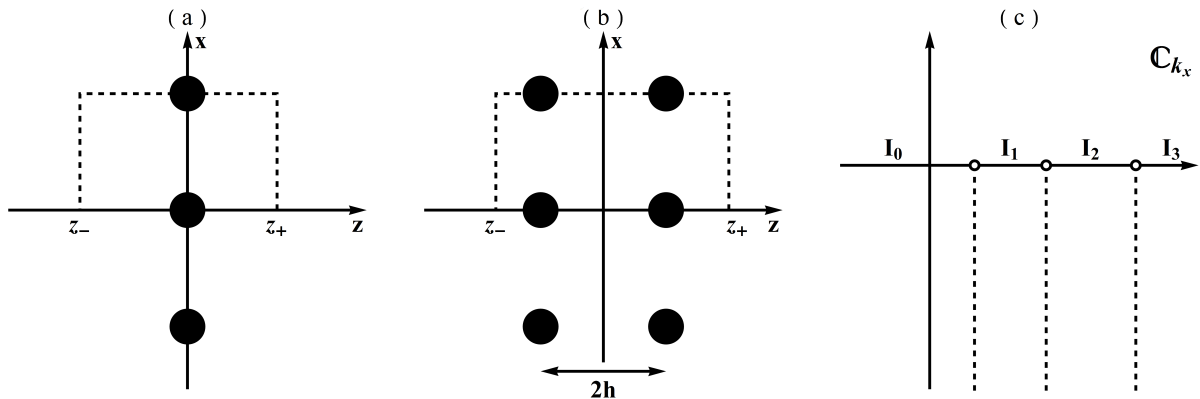}
    \caption{Panel (a): An example of
scattering structures considered in this work. Infinitely long parallel cylinders are placed parallel to a $y$-axis periodically along an $x$-axis in the vacuum $\ee_{0}=1$. The cylinders are characterized by a dielectric constant $\ee_{c}>1$. The rectangle $D$ in Eq.(\ref{eq:7}) is enclosed by the dashed line, and the $z$-axis. \newline
    Panel(b): Two periodic arrays such as the one on Panel (a) are placed parallel to each other at a distance $2h$ between the axes of any two opposing cylinders. \newline
    Panel (c): The $\ckx$ plane. The cuts run vertically from the diffraction thresholds (indicated by empty circles on the real axis) into the lower half of the complex plane. The diffraction thresholds divide the real axis into a countable set of intervals denoted $\ib_{l},\,l\geq 0$. The interval $\ib_{0}$ lies below the radiation continuum. The rest of the intervals partition the radiation continuum. }
    \label{fig:1}
    \end{figure}

The existence of solutions to the Lippmann-Schwinger integral equation that also satisfy Bloch's condition is established by extending the operator valued function $k^2\mapsto \hv(k^2)$ to a suitably cut complex plane, and thereafter, applying the Fredholm theory of compact operators.
To this end, let the $x-$component $\kx$ of the incident wave vector $\kk$ be fixed. The magnitude $k=|\kk|$ varies
only in its $z-$component,
$\kz=\pm \sqrt{k^2-k_{x}^{2}}$. Now let $S_{\ee}$ be the support of the function $(x,z)\rightarrow (\ee(x,z)-1)$ in the strip $S=[0,1]\times (-\infty,\infty)$ of the $x,z$-plane, and let $D$ be a rectangle in $S$ containing $S_{\ee}$, i.e.,
\begin{equation}
S_{\ee}\subset D=[0,1]\times [z_{-},z_{+}]
\label{eq:7}
\end{equation}
Then for an amplitude $\eo$ satisfying Bloch's condition,
\begin{equation}
\hv[\eo](\rr)=\frac{k^2}{4\pi}\sum_{m\in {\mathbbm{Z}}}e^{i m\kx}
\int_{D} (\ee(\ro)-1)G_{k}(\rr|\ro-m\ex)\eo(\ro) d\ro
\label{eq:8}
\end{equation}
where $\mathbbm{Z}$ denotes the set of all integers. This equation defines $\hv$ as an operator on $L^{2}(D)$, the Hilbert space of square integrable functions on $D$ with respect to the Lebesgue measure. Note that $\hv$ does not depend on the rectangle $D$ which to some extent is arbitrary. Indeed, the integrals of Eq.(\ref{eq:8}) extend only over the support $S_{\ee}$ of the function $(x,z)\mapsto (\ee(x,z)-1)$ in the strip $S=[0,1]\times (-\infty,\infty)$ of the $x,z$-plane. The rectangle $D$ is only introduced to obtain a connected and compact region of integration, which
is convenient for the subsequent analysis.
In general, the support of the function
$(x,z)\mapsto (\ee(x,z)-1)$ is not necessarily
connected as, for example, in the case of multiple scatterers in the strip $S$ depicted in Fig.~\ref{fig:2}(b).
The solution to the Lippmann-Schwinger
integral equation (\ref{eq:5}) will be sought in the Hilbert space $L^{2}(D)$. Such a solution then extends naturally
to the full $x,z-$plane by Bloch's periodicity
condition (\ref{eq:4}), and the Lippmann-Schwinger integral equation.

It is not hard to see that the summation and integration can be interchanged in Eq.(\ref{eq:8})
even though the underlying series is only conditionally convergent. The Poisson summation formula is then applied to yield the following form for the integral operator $\hv$,
\begin{subequations}\label{eq:9}
\begin{equation}\label{eq:9.1}
\hv[E](\rr)=\int_{D} (\ee(\ro)-1)H(k^2;\rr-\ro)E(\ro) d\ro,
\quad E\in L^{2}(D)
 \end{equation}
where the function $H$ is given by,
 \begin{equation}\label{eq:9.2}
 H(\es;x,z)=\frac{i \es}{2}\sum_{m\in \mathbbm{Z}} \frac{e^{i(x k_{x,m}+|z|\sqrt{\es-k_{x,m}^2})}}{\sqrt{\es-k_{x,m}^2}},\quad k_{x,m}=k_{x}+2\pi m,\quad \es\neq k_{x,m}^{2},\quad
m\in\mathbbm{Z}
 \end{equation}
\end{subequations}
The square roots are defined by choosing the branch cut of the logarithm along the negative imaginary axis, i.e., if $w$ is a nonzero complex number, then
\begin{equation}
\log w=\ln |w|+i \arg w,\quad
\,-\frac{\pi}{2}<\arg w<\frac{3\pi}{2}
\label{eq:10}
\end{equation}
In the scattering theory, the branch points
$k_{x,m}^{2}$, $ m\in \mathbbm{Z}$,
are called the {\it diffraction thresholds}, and they will be referred to as such in what
follows.

Note that for $w$ real, then $\im\{\sqrt{w}\}\geq 0$. Therefore, $\im\left\{\sqrt{\es-\kmm}\right\}>0$ for all but a finite number of diffraction thresholds $k_{x,m}^{2}$ such that $\es=k^2>\kmm$. It follows that in the series of Eq.(\ref{eq:9.2}) all but a finite number of terms decay exponentially as $|z|\rightarrow\infty$. Hence the said series converges uniformly on $D$. The uniform convergence on $D$ still holds when the range of the variable $\es=k^2$ is extended to the cut complex plane $\ckx$ obtained by excluding all the vertical half lines running from the diffraction thresholds into the lower half of the complex plane, i.e.,
\begin{equation*}
\ckx=\mathbbm{C}-\bigcup_{m\in \mathbbm{Z}}\{k_{x,m}^{2}-i s,\, s\leq 0\}
\label{eq:10.11}
\end{equation*}
This is because in this case too, $\im\left\{\sqrt{\es-\kmm}\right\}>0$ except for a finite number of terms. It follows that $\es\mapsto H(\es;\rr)$ extends to an analytic function on $\ckx$. Figure~\ref{fig:1}(c) shows a sketch of the cut plane $\ckx$.

Since for all $\es\in \ckx$, the kernel $(\rr,\ro)\mapsto (\ee(\ro)-1)H(\es;\rr-\ro)$ is square integrable on $D\times D$, the integral operator $\hv(\es)$ is Hilbert-Schmidt in $\mathcal{L}\{L^{2}(D)\}$, the space of bounded linear operators on $L^{2}(D)$. In particular, $\hv(\es)$ is compact for each $\es\in \ckx$. Moreover, the operator valued function defined on $\ckx$ by $\es\mapsto \hv(\es)$ is analytic. By the analytic Fredholm theorem~\cite{b2}, it then follows that if the inverse operator $[1-\hv(\es)]^{-1}$ exists at some point $\es\in \ckx$, it must be meromorphic throughout $\ckx$. That the said inverse exists at some point is obvious. Indeed, as $\es\rightarrow 0$ in $\ckx$,
the norm of the operator $\hv(\es)$, i.e., $||\hv(\es)||=||(\ee(\cdot)-1)H(\es;\cdot)||_{L^{2}(D\times D)}$, also converges to zero. In particular for $\es$ near $0$, $||\hv(\es)||<1$ and therefore $[1-\hv(\es)]^{-1}$ exists as a Neumann series. Thus, the generalized resolvent $\es\mapsto [1-\hv(\es)]^{-1}$ is meromorphic in $\ckx$.

The poles $\{\es_{n}\}_{n}$ of $[1-\hv(\es)]^{-1}$ are isolated and form a discrete set. The same analytic Fredholm theorem guarantees that at each of the poles $\es_{n}$, there exists a nonzero solution to the generalized eigenvalue problem
$$
\hv(\es_{n})[E_{n}]=E_{n}
$$
The generalized eigenfunctions $E_n$
will be referred to as the {\it Siegert states}.

It will be assumed throughout the rest of the work that the Siegert states are nondegenerate, i.e., the poles $\es_{n}$ are all {\it simple}. The results to be derived in the rest of the work could easily be generalized to the case when
the poles are not simple. Yet, to our knowledge, there has been no report of resonance poles of multiplicity greater than one in either the literature of electromagnetism or that of quantum mechanics. However, we do not have a rigorous proof that higher multiplicity poles cannot occur
in the studied scattering problem.
For the structures depicted in Fig.~\ref{fig:1},
perturbation theory (when the radius of cylinders
is much less than the structures period)
shows that all the poles are simple
~\cite{b5,b6}.

The assumed nondegeneracy of the Siegert states $\{E_{n}\}_{n}$ implies that the residues $\{\hv_{n}\}_{n}$ of $[1-\hv(\es)]^{-1}$ at the poles $\{\es_{n}\}_{n}$ are rank one operators on $L^{2}(D)$. In other words, if $\langle f,g\rangle =\int_{D} \overline{f(\ro)} g(\ro) d\ro$ is the inner product on the Hilbert space $L^{2}(D)$, then
\begin{equation}
\forall n,\, \exists \varphi_{n}\in L^{2}(D): \forall \psi\in L^{2}(D),\quad \, \hv_{n}[\psi]=\langle \varphi_{n}, \psi\rangle E_{n}
\label{eq:11}
\end{equation}
Since the incident wave $E_{i}(k^2;\rr)=e^{i\kk\cdot\rr}$ is analytic in $\es=k^2$ on $\ckx$, it follows that the amplitude $\eo=[1-\hv(\es)]^{-1}[E_{i}(\es;\cdot)]$ extends to a meromorphic function of the variable
$\es=k^2$ on $\ckx$. Its partial fraction expansion reads
\begin{equation}
\eo(\rr)=e^{i\kk\cdot \rr}+\sum_{n}\frac{a_{n}}{k^{2}-k_{n}^{2}+i\Gamma_{n}}E_{n}(\rr)+E_{a}(k^2;\rr),\quad
\, a_{n}=\langle \varphi_{n}, E_{i}(k^{2};\cdot)\rangle \Big|_{k^2=k_{n}^{2}-i\Gamma_{n}}
\label{eq:12}
\end{equation}
where each pole $\es_{n}$ has been decomposed in its real and imaginary parts as $\es_{n}=k_{n}^{2}-i\Gamma_{n}$ owing to the fact that the imaginary parts of the poles are nonpositive as will be shown shortly. The remainder $E_{a}$ is analytic in $k^2$ on the cut plane $\ckx$. It describes the so called {\it background} or {\it potential scattering},
while the Siegert states account for the resonant scattering.

Since the Siegert states are solutions to the homogeneous linear equation $\hv(\es)[E]=E$, they should be normalized in some way in order for the functions $\varphi_{n}$ in Eq.(\ref{eq:11}) and the coefficients $a_{n}$ in Eq.(\ref{eq:12}) to be uniquely defined. A natural normalization
condition will be introduced in Section~\ref{sec:4} for Siegert states that arise as perturbations of bound states in the radiation continuum.

That all the poles of the generalized resolvent $\es\mapsto [1-\hv(\es)]^{-1}$ have nonpositive imaginary parts, i.e., $\Gamma_{n}\geq 0,\, \forall n$, follows from a series of observations. First, note that a Siegert state $E_{n}$
corresponding to the pole $\es_{n}=k^{2}_{n}-i\Gamma_{n}$ satisfies Eq.(\ref{eq:2}) for $k^2=\es_{n}$, and therefore
\begin{equation*}
\overline{E_{n}}\Delta E_{n}-E_{n}\Delta \overline{E_{n}}-2i\Gamma_{n} \ee |E_{n}|^2=0
\label{eq:13}
\end{equation*}
By Green's theorem, it follows that,
\begin{equation}
2i\Gamma_{n}\int_{D'}|E_{n}(\ro)|^2\ee(\ro) d\ro =  \oint_{\partial D'} (\overline{E_{n}(\ro)}\nabla E_{n}(\ro)-E_{n}(\ro)\nabla \overline{E_{n}(\ro)})\cdot d\mathbf{n_{0}}
\label{eq:14}
\end{equation}
where $D'=[0,1]\times[z_{1},z_{2}]$ is any rectangle containing the rectangle $D$, and $\mathbf{n_{0}}$ is the outward normal. By Bloch's periodicity condition, the integrals on the line segments $(x,z)\in \{0,1\}\times [z_{1},z_{2}]$ of $\partial D'$ are canceled out so that the integral on the right side is to be carried out
only over the line segments $(x,z)\in [0,1]\times \{z_{1},z_{2}\}$ of $\partial D'$. The result can be expressed in terms of the scattering amplitudes $S_{m}^{\pm}$
associated with the asymptotic behavior
of the Siegert state $E_{n}$:
\begin{subequations}\label{eq:15}
\begin{equation}\label{eq:15.1}
E_{n}(\rr)=\begin{cases}
           \displaystyle{\sum_{m\in\mathbbm{Z}} S_{m}^{-} e^{i(x k_{x,m}-z\sqrt{\es_{n}-k_{x,m}^{2}})},\quad \, z<z_{-}}\\
           \displaystyle{\sum_{m\in\mathbbm{Z}} S_{m}^{+} e^{i(x k_{x,m}+z\sqrt{\es_{n}-k_{x,m}^{2}})},\quad \, z>z_{+}}
           \end{cases}
           \end{equation}
Equations (\ref{eq:15.1}) follow immediately from equations (\ref{eq:9}). In particular, the amplitudes $S_{m}^{\pm}$ are given by the formulas,
\begin{equation}\label{eq:15.2}
S_{m}^{\pm}=\frac{i\es_{n}}{2\sqrt{\es_{n}-k_{x,m}^{2}}}\int_{D}(\ee(\ro)-1)E_{n}(\ro)e^{i(-x_{0} k_{x,m}\mp z_{0} \sqrt{\es_{n}-k_{x,m}^{2}})}d\ro,\quad\, \ro=x_{0}\ex+z_{0}\ez
\end{equation}
\end{subequations}
Evaluating the right hand side of Eq.(\ref{eq:14}) yields,
\begin{equation}
\Gamma_{n}=\frac{\sum_{m\in\mathbbm{Z}} \re\left\{\sqrt{\es_{n}-\kmm}\right\}\left(|S_{m}^{+}|^{2}e^{-2z_{2}\im\left\{\sqrt{\es_{n}-\kmm}\right\}}
+|S_{m}^{-}|^{2}e^{2z_{1}\im\left\{\sqrt{\es_{n}-\kmm}\right\}}\right)}{{\int_{D'}|E_{n}(\ro)|^{2}\ee(\ro) d\ro}}
\label{eq:16}
\end{equation}
The series in Eq.(\ref{eq:16}) is then split into the sums over the two complementary index sets,
\begin{subequations}\label{eq:17}
\begin{equation}\label{eq:17.1}
\ib^{op}(\es_{n})=\{m\in \mathbbm{Z}\ |\
 \re\{\es_{n}\}\geq \kmm\},\quad\,\ib^{cl}(\es_{n})=\{m\in \mathbbm{Z}\ |\ \re\{\es_{n}\}< \kmm\}
\end{equation}
In the scattering theory, when $\es=k^2$ is real, the sets $\ib^{op}(k^2)$ and $\ib^{cl}(k^2)$ label open and closed diffraction channels respectively; thus the superscripts ``$op$'' for open, and ``$cl$'' for closed.

The choice of the branch cut for the logarithm given in Eq.(\ref{eq:10}) ensures that
\begin{equation}\label{eq:17.2}
\im\left\{\sqrt{\es_{n}-\kmm}\right\}>0,\quad
 \forall m\in \ib^{cl}(\es_{n})
\end{equation}
\end{subequations}
In particular, in the series of Eq.(\ref{eq:16}), the contributions from terms whose indices lie in $\ib^{cl}(\es_{n})$ decay exponentially as $z_{2}\rightarrow +\infty$ and $z_{1}\rightarrow -\infty$. In the said limits, the aforementioned series is therefore reduced to a finite
sum over the set $\ib^{op}(\es_{n})$:
\begin{equation*}
\Gamma_{n}=\lim_{\begin{subarray}{l} z_{1}\rightarrow -\infty \\ z_{2}\rightarrow +\infty\end{subarray}}\frac{\sum_{m\in \ib^{op}(\es_{n})} \re\left\{\sqrt{\es_{n}-\kmm}\right\}\left(|S_{m}^{+}|^{2}e^{-2z_{2}\im\left\{\sqrt{\es_{n}-\kmm}\right\}}
+|S_{m}^{-}|^{2}e^{2z_{1}\im\left\{\sqrt{\es_{n}-\kmm}\right\}}\right)}{{\int_{D'}|E_{n}(\ro)|^{2}\ee(\ro) d\ro}}
\label{eq:18}
\end{equation*}
Observing that $\re\left\{\sqrt{\es_{n}-\kmm}\right\}\geq 0$ for all $m\in \ib^{op}(\es_{n})$, and since the dielectric function $\ee$ is positive, it follows that $\Gamma_{n}\geq 0$, and therefore the imaginary part of the pole $\es_{n}$ is indeed necessarily negative or zero, i.e., all the poles $\es_{n}$ are in the lower half of the cut complex plane $\ckx$.

\section{Regular perturbation theory}
\label{sec:3}

Suppose that the dielectric function $\ee$ depends on a real parameter $h$ which is called a {\it coupling
parameter}. The simplest example
of such a coupling is given by
two periodic arrays of dielectric scatterers
that are parallel and separated by
the distance $2h$ (Fig.~\ref{fig:1}(b)). The arrays are embedded in a medium of dielectric susceptibility $1$. Each array is characterized by a dielectric function $\ee_{i}(x,z),\, i=1,2$ such that $\ee_{i}\geq 1$ on the scatterers, and $\ee_{i}$ is of period $1$ in the $x$-direction. The resulting dielectric function describing the scattering of light on the structure is then,
\begin{equation}
\ee(h;x,z)=1+(\ee_{1}(x,z-h)-1)+(\ee_{2}(x,z+h)-1)
\label{eq:19}
\end{equation}

In the most general case, the dependence of the dielectric function $\ee$ on the coupling parameter $h$ implies that the poles of the generalized resolvent $\es\mapsto [1-\hv(h,\es)]^{-1}$ also depend on $h$. If the coupling is sufficiently smooth as indicated in the theorem stated below, then the poles of the generalized resolvent as well as the corresponding Siegert states depend continuously on $h$.

Before stating the theorem, a clarification must be made on the boundary of the rectangle $D$ in Eq.(\ref{eq:7}). This rectangle was chosen to contain the support $S_{\ee}$ of the function $(x,z)\mapsto (\ee(x,z)-1)$ in the strip $S=[0,1]\times(-\infty,\infty)$ of the $x,z$-plane. In the current situation, the dielectric function depends on the coupling parameter $h$, and therefore the support $S_{\ee(h)}$ in question could change with $h$ resulting in a different choice for the set $D$. This is the case, for instance, in the example of the two parallel arrays separated by
the distance $2h$. If the two scatterers in the strip $S$ are taken further apart, the rectangle $D$ is stretched further
accordingly. So, in what follows it will always be assumed that $h$ varies in an open possibly finite interval $J_{0}$ such that $\bigcup_{h\in J_{0}} S_{\ee(h)}$ is bounded, and the rectangle $D$ will be chosen to contain the latter set. The Theorem on the regular perturbation of electromagnetic Siegert states is
formulated as follows:

\begin{rpth}
Suppose that the map $(h,\es)\mapsto\hv(h,\es)$ from $J_{0}\times\ckx$ to $\mathcal{L}\{L^{2}(D)\}$ is continuously Frechet differentiable. Further, suppose that for some $h_{n}\in J_{0}$, the generalized resolvent $\es\mapsto [1-\hv(h_{n},\es)]^{-1}$ has a simple pole $\es_{n}\in\ckx$, and let $E_{n}$ be the corresponding Siegert state. Then there exists an open interval $J\subset J_{0}$ containing $h_{n}$, and a unique continuously differentiable function $h\mapsto \ehf(h),\, h\in J$, such that $\ehf(h_{n})=\es_{n}$, and for all $h\in J$, $\ehf(h)$ is a simple pole of the generalized resolvent $\es\rightarrow [1-\hv(h,\es)]^{-1}$. If $\sih(h)$ is the Siegert state corresponding to the pole $\ehf(h)$, then the function $h\mapsto \sih(h)$ from $J$ to $L^{2}(D)$ is continuously differentiable.
\end{rpth}

This theorem is an extension of the Kato-Rellich theorem~\cite{b3} to the present case. Due to its length and complexity, the proof is left to the Appendix. It should be noted, however, that if the functions $(x,z)\mapsto \ee_{i}(x,z),\,i=1,2$ in Eq.(\ref{eq:19}) are piecewise differentiable, then $(h,\es)\mapsto\hv(h,\es)$ is Frechet differentiable so that the theorem does indeed hold for two parallel arrays
separated by the distance $h$.

\subsection{Perturbation of bound states in the radiation continuum}
\label{sec:4}

The diffraction thresholds $k_{x,m}^{2}$ on the real line depend on the $x-$component $\kx$ of the incident wave vector.
They can be ordered independently of the parameter $\kx$ by defining a sequence $\es_{\pm m}^{*}(\kx)=(2\pi m\pm \left|[k_{x}]\right|)^{2},\, m\in \{0,1,2,\ldots\}$, where $[k_{x}]$ is the argument of $e^{i\kx}$ in the interval $(-\pi,\pi]$. The sequence $\{\es_{\pm m}^{*}(\kx)\}_{m=0}^{\infty}$ coincides with the sequence of diffraction thresholds $\{(\kx+2\pi m)^{2}\}_{m\in\mathbbm{Z}}$, and for all $\kx$,
\begin{equation*}
\es_{0}^{*}(\kx)\leq\es_{-1}^{*}(\kx)\leq\es_{1}^{*}(\kx)\leq\es_{-2}^{*}(\kx)\leq\es_{2}^{*}(\kx)\leq\es_{-3}^{*}(\kx)\leq\ldots
\label{eq:20}
\end{equation*}

The threshold $\es_{0}^{*}(\kx)$ is called {\it the radiation continuum threshold}, and the interval
$\ib_{0}=(-\infty,\es_{0}^{*}(\kx))$ is said to lie {\it below the radiation continuum}.
This is because whenever $k^2<\es_{0}^{*}(\kx)$, then the scattered amplitude $E_{\omega}^{s}$ in Eq.(\ref{eq:3}) necessarily decays exponentially in the spatial infinity, $|z|\rightarrow\infty$, as can be inferred from Eqs.(\ref{eq:3}), (\ref{eq:5}), and (\ref{eq:9}). Hence, no electromagnetic
flux is carried to the spatial infinity in this spectral
range. In contrast,
if $k^2>\es_{0}^{*}(\kx)$, the amplitude $E_{\omega}^{s}$
oscillates and represents outgoing (scattered) radiation
that carries an electromagnetic flux to
the spatial infinity.
The spectral range above $\es_{0}^{*}(\kx)$
on a real line is therefore referred to as {\it the radiation continuum}. It is the disjoint union of the intervals,
\begin{equation*}
\ib_{1}=(\es_{0}^{*}(\kx),\es_{-1}^{*}(\kx)),\quad
\ib_{2}=(\es_{-1}^{*}(\kx),\es_{1}^{*}(\kx)),\quad
\ib_{3}=(\es_{1}^{*}(\kx),\es_{-2}^{*}(\kx)),\quad\ldots
\label{eq:21}
\end{equation*}
Note that when $[\kx]$ is $0$ or $\pi$, some of the intervals $\ib_{l}$ are empty, owing to the fact that some of the diffraction thresholds fuse.

Returning to Siegert states, recall that these states are generalized eigenfunctions to the generalized eigenvalue problem $\hv(\es_{n})[E_{n}]=E_{n}$ on $L^{2}(D)$ with complex eigenvalues $\es_{n}$. These states are naturally
extended to the whole $x,z-$plane by means of
Eqs.~(\ref{eq:15}). In general, if the generalized eigenvalue $\es_{n}=k_{n}^{2}-i\Gamma_{n}$ lies below or in the interval $\ib_{0}$ in the $\ckx$-plane, then the corresponding Siegert state is square integrable on the strip $S=[0,1]\times(-\infty,\infty)$ of the $x,z$-plane. This is because for such states, the set $\ib^{op}(\es_{n})$ of Eq.(\ref{eq:17.1}) is empty, and therefore by Eqs.(\ref{eq:15.1}) and (\ref{eq:17.2}), these states decay exponentially in the asymptotic region $|z|\rightarrow\infty$. In particular, the Siegert states for which the pole $\es_{b}$ is real and less than the continuum threshold $\es_{0}^{*}(k_{x})$ are the {\it bound states below the radiation continuum} of the system.

On the contrary, the Siegert states $E_{n}$ whose corresponding generalized eigenvalues $\es_{n}=k^{2}_{n}-i\Gamma_{n}$ lie below an interval $\ib_{l},\, l\geq 1$ of the radiation continuum, i.e., $k^{2}_{n}\in \ib_{l}$ and $\Gamma_{n}>0$, are not necessarily square integrable on the strip $S$. This is because for these states, the set $\ib^{op}(\es_{n})$ is not empty, and consequently, the terms of the series in Eq.~(\ref{eq:15.1}) indexed by this set are unbounded. However, when the pole $\es_{n}$ is in $\ib_{l},\, l\geq 1$, i.e., $\Gamma_{n}=0$, it will be shown shortly that the resulting Siegert state is square integrable on the strip $S$ and therefore, such a Siegert state is a {\it bound state in the radiation continuum}. The rest of this subsection
is devoted to the proof of this assertion.

To proceed, suppose that for some value $h_{b}$ of the coupling constant $h$, there exists a Siegert state $E_{b}$ whose corresponding generalized eigenvalue $\es_{b}$ is real and lies in an interval $\ib_{l},\, l\geq 1,$ of the radiation continuum. By the Regular Perturbation Theorem there exist continuously differentiable functions $h\mapsto \ehf(h)$ and $h\mapsto \sih(h)$ on an interval $J$ containing $h_{b}$ such that $\ehf(h_{b})=\es_{b}$, and $\sih(h_{b})=E_{b}$. Furthermore, $\sih(h)$ is the Siegert state corresponding to the pole $\ehf(h)$ of the generalized resolvent $\es\mapsto [1-\hv(h,\es)]^{-1}$ for all $h\in J$. Put $\ehf(h)=k_{n}^{2}(h)-i\Gamma_{n}(h)$ as before. Then $\Gamma_{n}(h_{b})=0$. Without loss of generality, the interval $J$
is assumed to be
sufficiently small so that for
all $h\in J\backslash\{h_{b}\}$, $\Gamma_{n}(h)\neq 0$.
Equation~(\ref{eq:16}) for the Siegert state $\sih(h)$ is then rewritten as
\begin{equation*}
\int_{D'}|E_{n}(\ro)|^{2}\ee(h;\ro) d\ro=\frac{1}{\Gamma_{n}}\sum_{m\in\mathbbm{Z}} \re\left\{\sqrt{\ehf-\kmm}\right\}\left(|S_{m}^{+}|^{2}e^{-2z_{2}\im\left\{\sqrt{\ehf-\kmm}\right\}}
+|S_{m}^{-}|^{2}e^{2z_{1}\im\left\{\sqrt{\ehf-\kmm}\right\}}\right)
\label{eq:22}
\end{equation*}
where it is understood that the Siegert state $E_{n}$, the pole $\ehf$, and its imaginary part $-i\Gamma_{n}$, as well as the amplitudes $S_{m}^{\pm}$ are all functions of $h$. Note that the values of $z_{1}$ and $z_{2}$ are assumed to be independent of the coupling parameter $h$. This is because $h$ varies in a small interval $J$, and therefore the values of $z_{\pm}$ in Eq.(\ref{eq:7}), which now depend on $h$ are bounded. The condition $z_{1}\leq z_{-}(h)<z_{+}(h)\leq z_{2},\, \forall h\in J$, can therefore be realized by choosing $z_{1}$ and $z_{2}$ sufficiently large.

As $h\rightarrow h_{b}$, then $E_{n}\rightarrow E_{b}$ in $L^{2}(D)$, and therefore the left hand side of Eq.(\ref{eq:12}) remains finite as the dielectric function $\ee(h;\cdot)$ is bounded. The limit of the right hand side as $h\rightarrow h_{b}$ may also be calculated by first computing explicitly the complex square roots involved according to the logarithmic branch cut described in Eq.(\ref{eq:10}).
Put $\xi_{m}(h)=k_{n}^{2}(h)-\kmm,\, m\in\mathbbm{Z}$. Then
\begin{equation*}
\sqrt{\eh(h)-\kmm}=\begin{cases}
                   \sqrt{\frac{\sqrt{\xi_{m}(h)^{2}+\Gamma_{n}(h)^2}+\xi_{m}(h)}{2}}-\frac{i\Gamma_{n}(h)}{\sqrt{2\left(\sqrt{\xi_{m}(h)^{2}+\Gamma_{n}(h)^2}+\xi_{m}(h)\right)}},\,&\text{if $m\in \ib^{op}(\eh(h))$}\\
                   -\frac{\Gamma_{n}(h)}{\sqrt{2\left(\sqrt{\xi_{m}(h)^{2}+\Gamma_{n}(h)^2}-\xi_{m}(h)\right)}}+i\sqrt{\frac{\sqrt{\xi_{m}(h)^{2}+\Gamma_{n}(h)^2}-\xi_{m}(h)}{2}},\,&\text{if $m\in \ib^{cl}(\eh(h))$}
                   \end{cases}
                   \label{eq:23}
                   \end{equation*}
where $\ib^{op}$ and $\ib^{cl}$ are the index sets of Eq.(\ref{eq:17.1}). In particular, if $m\in \ib^{op}(\ehf(h))$, then $\xi_{m}(h)\geq 0$, whereas $\xi_{m}(h)<0$ whenever $m\in \ib^{cl}(\ehf(h))$. As $h\rightarrow h_{b}$, $\Gamma_{n}(h)\rightarrow 0$, and,
\begin{equation}
\begin{split}
\int_{D'}|E_{b}(\ro)|^2\ee(h_{b};\ro)d\ro=\lim_{h\rightarrow h_{b}}&\left[\frac{1}{\Gamma_{n}(h)}\sum_{m\in
\ib^{op}(\es_{b})}\sqrt{\xi_{m}(h)}\left(|S_{m}^{+}(h)|^2+|S_{m}^{-}|^2\right)\right]\\
&-\sum_{m\in\ib^{cl}(\es_{b})}\frac{1}{2 \sqrt{-\xi_{m}(h_{b})}}\left(|S_{m}^{+}(h_{b})|^2 e^{-2 z_{2}\sqrt{-\xi_{m}(h_{b})}}+|S_{m}^{-}(h_{b})|^{2}e^{2z_{1} \sqrt{-\xi_{m}(h_{b})}}\right)
\end{split}
\label{eq:24}
\end{equation}
In particular, for each $m\in \ib^{op}(\es_{b})$, the limits $\lim_{h\rightarrow h_{b}} \frac{|S_{m}^{\pm}(h)|^2}{\Gamma_{n}(h)}$ must exist and must be finite. By continuity of
the functions $h\mapsto S^{\pm}_{m}(h)$ on $J$,
there exist complex numbers $S_{m,b}^{\pm}$ such that
\begin{equation}
\lim_{h\rightarrow h_{b}} \frac{S_{m}^{\pm}(h)}{\sqrt{\Gamma_{n}(h)}}=S_{m,b}^{\pm},\quad\forall m\in \ib^{op}(\es_{b})
\label{eq:25}
\end{equation}

As $z_{1}\rightarrow -\infty$ and $z_{2}\rightarrow \infty$, the second series in Eq.(\ref{eq:24}) converges to zero.  Therefore the function $\rr\rightarrow E_{b}(\rr)\sqrt{\ee(h_{b};\rr)}$ is square integrable on the strip $S=[0,1]\times (-\infty,\infty)$ of the $x,z-$plane, and
\begin{subequations}\label{eq:26}
\begin{equation}\label{eq:26.1}
\int_{S} |E_{b}(\ro)|^{2}\ee(h_{b};\ro) d\ro=\sum_{m\in\ib^{op}(\es_{b})}\sqrt{\xi_{m}(h_{b})}\left(|S_{m,b}^{+}|^2+|S_{m,b}^{-}|^{2}\right)
\end{equation}
Since, $\ee(h_{b},\cdot)\geq 1$, it follows that $E_{b}\in L^{2}(S)$ as claimed. Thus, if $E_{b}$ is a bound state in the radiation continuum, it can always be normalized
by the condition
\begin{equation}\label{eq:26.2}
\int_{S} |E_{b}(\ro)|^{2}\ee(h_{b};\ro) d\ro=1
\end{equation}
\end{subequations}
This result justifies the term "bound state"
introduced by analogy
with quantum mechanics where bound states represent
square integrable eigenfunctions of the Hamiltonian
operator, i.e., an electromagnetic bound state
is a localized solution of Maxwell's equations.

Finally, by continuity $E_{n}(h)\rightarrow E_b$
as $h\rightarrow h_b$, and therefore the normalization condition (\ref{eq:26.2}) determines uniquely the Siegert states $E_{n}(h)$ along the curve $h\mapsto \ehf(h),\, h\in J$.

\subsection{Near field amplification mechanism}
\label{sec:5}

Resonant phenomena in the scattered electromagnetic flux
can be described by the formalism of
Siegert states as given in Eq.(\ref{eq:12}).
However, a complete description requires calculating
the residues $a_{n}$ in Eq.(\ref{eq:12}). Here
these residues are calculated for
Siegert states that arise as perturbations of bound states in the radiation continuum in the sense of Section~\ref{sec:4},
i.e., for the states $E_n(h)$ where $|h-h_b|/h_b\ll 1$.
In particular, it is shown that if the scattering structure
has bound states in the radiation continuum, then
there exist regions (the so called "hot spots") in which
the field amplitude can be amplified as much as desired
by taking the value of the coupling parameter $h$
close enough to $h_b$. The unbounded local growth of the field
amplitude is essentially due to the linearity of Maxwell's
equations. It disappears when a non-linear dielectric
susceptibility, required for large field amplitudes,
is included into Maxwell's equations~\cite{shg}.
Nevertheless, such a local field amplification
enhances quite substantially optical nonlinear effects
in the scattering structure as shown in \cite{shg}.
The following analysis also suggests that the concept to
enhance optically non-linear effects by using
bound states in the radiation continuum is universal
because the existence of "hot spots" is proved to be
a characteristic feature of such scattering structures.

When estimating the coefficients $a_{n}$, it is convenient
first to give another equation for them
that is alternative to that of Eq.(\ref{eq:12}).
Next, this equation  will be analyzed in the vicinity of a bound state in the radiation continuum. In particular, if a bound state in the radiation continuum $E_{b}$ exists at the critical value $h_{b}$ of the coupling parameter, and $J$ is the interval produced by the Regular Perturbation Theorem, while $E_{n}(h;\rr)$ is the Siegert state that arises as a continuous perturbation of the bound state $E_{b}$ for $h\in J$, then it will be proved that $a_{n}(h)\sim\sqrt{\Gamma_{n}(h)}$ as $h\rightarrow h_{b}$. Recall that $\ehf(h)=k_{n}^{2}(h)-i\Gamma_{n}(h)$ is the generalized eigenvalue at which the Siegert state $E_{n}(h)$ exists.

To proceed, put $\eot=\frac{\es-\ehf}{a_{n}}\eo$ where $\eo$ is the amplitude in Eq.(\ref{eq:3}). Then by the decomposition
(\ref{eq:12}),
$\eot\rightarrow E_{n}$ as $\es\rightarrow \es_{n}$. From the system,
\begin{equation*}
\begin{cases}
\Delta \eot+\es\ee(h) \eot=0\\
\Delta E_{n}+\es\ee(h) E_{n}=0
\end{cases}
\label{eq:27}
\end{equation*}
it is derived by Green's theorem that,
\begin{equation*}
\int_{\partial D'}\left(\overline{E_{n}}\nabla \eot-\eot \nabla \overline{E_{n}}\right)\cdot d\mathbf{n_{0}}+(\es-\overline{\es_{n}})\int_{D'} \eot \overline{E_{n}}\ee(h)d\ro=0
\label{eq:28}
\end{equation*}
where $D'$ is the same rectangle as in Eq.~(\ref{eq:14}). Then, by splitting the amplitude $\eot$ in terms of the incident
and scattered waves, $\eot(\rr)=\frac{\es-\es_{n}}{a_{n}}e^{i\kk\cdot\rr}+\eos(\rr)$, one infers that
\begin{equation*}
a_{n}(h)=-\frac{(\es-\es_{n})\left[\int_{\partial D'}(i \kk \overline{E_{n}}-\nabla \overline{E_{n}})e^{i\kk\cdot\rr_{0}}\cdot d\mathbf{n}+(\es-\es_{n})\int_{D'} \overline{E_{n}}e^{i\kk\cdot\ro}\ee(h)d\ro\right]}{\int_{\partial D'} \left(\overline{E_{n}}\nabla \eos-\eos\nabla \overline{E_{n}}\right)\cdot d\mathbf{n}+(\es-\es_{n})\int_{D'} \eos\overline{E_{n}}\ee(h)d\ro}
\label{eq:29}
\end{equation*}
where it is understood that $E_{n}$ and $\es_{n}$ are functions of $h$. Since $a_{n}$ is independent of $\es$, the
desired expression for $a_{n}$ is obtained by taking
 the limit $\es\rightarrow \es_{n}$:
\begin{equation}
a_{n}(h)=-\frac{\int_{\partial D'} \left(i\kk_{n} \overline{E_{n}}-\nabla \overline{E_{n}}\right)e^{i\kk_{n}\cdot\ro}\cdot d\mathbf{n}}{\int_{\partial D'}\left(\overline{E_{n}}\nabla\partial_{\es}\eos-\partial_{\es}\eos\nabla \overline{E_{n}}\right)\cdot d\mathbf{n}\Big|_{\es=\es_{n}}+\int_{D'} |E_{n}|^2\ee(h) d\ro},\quad\, \kk_{n}=\kk\Big|_{\es=\es_{n}}
\label{eq:30}
\end{equation}
where for the first term in the denominator, l'H\^{o}pital's rule has been applied. This formula would not generally be useful as the term $\eos$ involves the coefficients $a_{n}$ implicitly. However, if the Siegert state $E_{n}(h)$ is
taken near a bound state in the radiation continuum, as  assumed here, the first-order perturbative
expression of $a_{n}$ only depends on the Siegert state $E_{n}(h)$.

Such a perturbative expression is obtained by analyzing each of the integrals in Eq.~(\ref{eq:30}) separately.
 First, it is observed that as $h\rightarrow h_{b}$, then $E_{n}(h)\rightarrow E_{b}$. By letting $z_{1}\rightarrow -\infty$, and $z_{2}\rightarrow\infty$, it follows from the normalization of Eq.(\ref{eq:26.2}) that the second integral in the denominator of Eq.(\ref{eq:30}) can be made arbitrarily close to 1 provided $h$ is sufficiently close to $h_{b}$. Hence, if the first integral of the said denominator can be shown to converge to zero in the limits considered, this would imply that in the leading order of perturbation theory, the denominator of Eq.~(\ref{eq:30}) is approximated by 1. That the first integral of the said denominator does indeed converge to zero as $h\rightarrow h_{b}$, $z_{1}\rightarrow-\infty$, and $z_{2}\rightarrow\infty$ can be established through a lengthy, but
straightforward calculation whose details are omitted. It can be carried out along the following lines. As noted earlier, Bloch's condition implies that the first integral in the denominator of Eq.~(\ref{eq:30}) is to be carried out on the segments $(x,z)\in [0,1]\times \{z_{1},z_{2}\}$ of the boundary of $D'$. Now, as $h\rightarrow h_{b}$, the Siegert state $E_{n}(h)$ becomes a bound state $E_{b}$, and therefore, it decays exponentially in the spatial infinity $|z|\rightarrow\infty$. Similarly, the same exponential decay can be established for the terms involving the derivatives of $\eos$ in the limit $h\rightarrow \hb$. It then follows that in the limits $h\rightarrow\hb$, $z_{1}\rightarrow-\infty$ and $z_{2}\rightarrow\infty$, the integral in question vanishes. Thus the denominator of Eq.(\ref{eq:30}) remains indeed close to 1, provided $h$ is sufficiently close to $h_{b}$.

On the other hand, the numerator to Eq.(\ref{eq:30}) can be evaluated in terms of the amplitudes $S_{m}^{\pm}$
in Eqs.(\ref{eq:15}) to yield
\begin{equation*}\label{eq:31.1}
\int_{D'} \left(i\kk_{n} \overline{E_{n}}-\nabla \overline{E_{n}}\right) e^{i\kk_{n}\cdot\ro} \cdot d\mathbf{n}= \sqrt{\Gamma_{n}(h)}A_{n}(h,z_{1},z_{2})
\end{equation*}
where if $\mathbf{k}_{n}=\kx\ex+\sqrt{\ehf-k_{x}^{2}}\ez$, then
\begin{equation*}\label{eq:31.2}
A_{n}(h,z_{1},z_{2})=2i\re\left\{\sqrt{\es_{n}(h)-k_{x}^{2}}\right\} \frac{\overline{S_{0}^{+}}(h)}{\sqrt{\Gamma_{n}(h)}} e^{-2 z_{2} \im\left\{\sqrt{\es_{n}(h)-k_{x}^{2}}\right\}}+2\im\left\{\sqrt{\es_{n}(h)-k_{x}^{2}}\right\} \frac{\overline{S_{0}^{-}}(h)}{\sqrt{\Gamma_{n}(h)}}e^{2i z_{1}\re\left\{\sqrt{\es_{n}(h)-k_{x}^{2}}\right\}}
\end{equation*}
The case $\mathbf{k}_{n}=\kx\ex-\sqrt{\ehf-k_{x}^{2}}\ez$ is similar. As $h\rightarrow h_{b}$, then $\es_{n}(h)\rightarrow \es_{b}$, and since $\es_{b}>k_{x}^{2}$, it follows that $\im\left\{\sqrt{\es_{n}(h)-k_{x}^{2}}\right\}\rightarrow 0$. Hence, in the said limit, $A_{n}(h,z_{1},z_{2})\rightarrow 2i \overline{S_{0,b}^{+}} \sqrt{\es_{b}-k_{x}^{2}}$ for $S_{0,b}^{+}$ given in Eq.(\ref{eq:25}). Thus, in general, $a_{n}$ can be written as
\begin{equation}
a_{n}(h)=\widetilde{a}_{n}(h) \sqrt{\Gamma_{n}(h)}
\label{eq:32}
\end{equation}
where $\widetilde{a}_{n}(h)$ is bounded and has the property $\widetilde{a}_{n}(h)\rightarrow 2i \overline{S_{0,b}^{+}}\sqrt{\es_{b}-k_{x}^{2}}$
as $h\rightarrow h_{b}$.

The most important consequence of the structure of the coefficients $a_{n}$ near a bound state in the radiation continuum is a local amplification of the amplitude $\eo$
as compared to the amplitude of the incident radiation. Indeed, if the wavenumber $k$ of the incident radiation and the coupling parameter $h$ are tuned to satisfy the condition $k=k_{n}(h),\,\forall h\in J$, then the amplitude
(\ref{eq:12}) becomes:
\begin{equation}
\eo(\rr)=i\frac{\widetilde{a}_{n}(h)}{\sqrt{\Gamma_{n}(h)}}E_{n}(\rr)+\widetilde{E}_{a}(h;\rr)
\label{eq:33}
\end{equation}
where $h\mapsto ||\widetilde{E}_{a}(h,\cdot)||_{L^{2}(D)}$ is bounded on $J$. It follows that,
\begin{equation*}
||\eo||_{L^{2}(D)}\geq\left| \frac{|\widetilde{a}_{n}(h)|}{\sqrt{\Gamma_{n}(h)}} ||E_{n}||_{L^{2}(D)}-||\widetilde{E}_{a}||_{L^{2}(D)}\right|
\label{eq:34}
\end{equation*}
As $h\rightarrow h_{b}$, then $E_{n}\rightarrow E_{b}$ in $L^{2}(D)$, and therefore $||E_{n}||_{L^{2}(D)}\rightarrow ||E_{b}||_{L^{2}(D)}\neq 0$. Hence if the constant $S_{0,b}^{+}$ in Eq.(\ref{eq:32}) is nonzero, then,
\begin{equation*}
\lim_{h\rightarrow h_{b}} ||\eo||_{L^{2}(D)}=\infty
\label{eq:35}
\end{equation*}
and this can only happen if the amplitude $\eo$ diverges in some regions of the rectangle $D$. Even though the fact that
$\widetilde{a}_{n}(h)$ is always nonzero in the limit
$h\rightarrow h_b$ (i.e.,
$S_{0,b}^{+}\neq 0$) could not be verified for all kinds of couplings, there are many systems in which it holds true.
For example,
this is the case of the normal incidence ($k_{x}=0$) for a symmetric double array depicted in Fig.~\ref{fig:1}(b) when the dielectric functions $\ee_{1}$ and $\ee_{2}$ in Eq.(\ref{eq:19}) are identical, and symmetric with respect to the reflection $(x,z)\mapsto (x,-z)$. In this particular case, the parity operator, $\widehat{\text{S}}[E](x,z)=E(x,-z)$ on $L^{2}(D)$, and the Lippmann-Schwinger integral operator $\hv$ commute so that Siegert states are always symmetric or skew symmetric with respect to the reflection $(x,z)\mapsto (x,-z)$ in the $x,z$-plane. In particular, the amplitudes $S_{m}^{\pm}$ given in Eqs.(\ref{eq:15}) are such that $S_{m}^{+}=\pm S_{m}^{-}$ depending on whether the Siegert state they correspond to is symmetric or skew symmetric. It follows that if the bound state $E_{b}$ happens to lie in the interval $\ib_{1}$ of the $\ckx$-plane (case of one open diffraction channel), then Eqs.(\ref{eq:26}) are reduced to
\begin{equation*}
|S_{0,b}^{+}|^2=|S_{0,b}^{-}|^{2}=\frac{1}{2\sqrt{\es_{b}-k_{x}^{2}}}
\label{eq:36}
\end{equation*}
so that the coefficient in question is indeed nonzero and the amplitude $\eo$ is amplified in the vicinity of the bound state in the radiation continuum $E_{b}$. Note that the amplification of the amplitude $\eo$ was also observed perturbatively in the case of a more general incidence angle (i.e., $k_{x}\neq 0$) when the bound state $E_{b}$ lies in the intervals $\ib_{1}$ and $\ib_{2}$ of the radiation continuum~\cite{b6}.

The aforementioned amplification can only happen in a region near or within the scattering structure. This can be understood by analyzing the first term of Eq.(\ref{eq:33}) from which the amplification should result. Outside the scattering region, the Siegert state is expressed in terms of its scattered amplitudes by Eq.(\ref{eq:15.1}). As before, the series involved in this expressions split over the index sets $\ib^{op}(\es_{b})$ and $\ib^{cl}(\es_{b})$ ( for $h$ near $h_{b}$, $\ib^{op}(\es_{n}(h))=\ib^{op}(\es_{b})$ and $\ib^{cl}(\es_{n}(h))=\ib^{cl}(\es_{b})$). On one hand, the terms of the series whose indices lie in $\ib^{cl}(\es_{b})$ decay exponentially in the spatial infinity so that no amplification of the field can be obtained from them. On the other hand, the terms whose indices lie in $\ib^{op}(\es_{b})$, disappear near the bound state in the radiation continuum as these terms are proportional to $\sqrt{\Gamma_{n}(h)}$ as indicated by Eq.(\ref{eq:25}). From the physical point of view,
the noted local field amplification results from
a constructive interference of scattered fields
from each scatterer (an elementary cell) of the periodic
structure. It is important to note that the amplification
magnitude can be regulated by varying the coupling
parameter $h$, which provides a natural physical mechanism
to control optical nonlinear phenomena if the structure
has a nonlinear dielectric susceptibility. This mechanism
has been used to demonstrate that a double
periodic array of dielectric cylinders (depicted
in Fig.~\ref{fig:1} (b)) with a non-zero
second-order nonlinear susceptibility can convert as much
as 44\% of the incident flux into the second harmonics
when the distance between the arrays is properly adjusted.
The smallest distance at which such a high conversion rate
can be achieved is about a half of the wave length of the
incident radiation~\cite{shg}.

\section{Decay of electromagnetic Siegert states}
\label{sec:6}

A Siegert state can be excited by an incident radiation, e.g.,
by a wave packet whose spectral distribution peaks around
a desired wave number $\kc\approx k_n$.
When passing through the structure, some of the wave packet energy
is trapped by the structure and remains in there for a long
period of time, decaying slowly by emitting a monochromatic radiation. This is established in a fashion similar to that of the study of the decay of unstable states in quantum mechanics~\cite{b4}. To illustrate the principle, only the case of normal incidence is considered here (i.e., $\kx=0$). Other incidence angles can be treated in a similar manner.

The exact statement which will be proved is as follows.
Suppose that a Siegert state $E_{n}$ exists at
the pole $\es_{n}=k_{n}^2-i\Gamma_{n}$ such that $k_{n}>0$
and $\Gamma_n>0$ (i.e., the Siegert state is
not a bound state in the radiation continuum).
Then this state
will decay with time by emitting a monochromatic radiation at a wavenumber $\widetilde{k}_{n}$, and with a
life time $\tau_{n}$:
\begin{equation}
\widetilde{k}_{n}=\sqrt{\frac{k_{n}^{2}+\sqrt{k_{n}^{4}+\Gamma_{n}^{2}}}{2}}\xrightarrow[\Gamma_{n}\ll 1]{} k_{n},\quad\,\, \tau_{n}=\frac{2\widetilde{k}_{n}}{c \Gamma_{n}} \xrightarrow[\Gamma_{n}\ll 1]{} \frac{2k_{n}}{c\Gamma_{n}}
\label{eq:37}
\end{equation}

As a starting point, consider an incident wave packet,
\begin{equation*}
E_{i}(\rr,t)=\int_{0}^{\infty} A(k) \cos(\kk\cdot\rr-\omega t)dk,\quad\, \omega=c k
\label{eq:38}
\end{equation*}
where $A(k)$ is the distribution of wavenumbers in the wave packet. Then the solution $E$ of the wave equation (\ref{eq:1}) reads
\begin{equation*}
E(\rr,t)=\re\left\{\int_{0}^{\infty} A(k)\eo(\rr)e^{-i \omega t} dk \right\}
\label{eq:39}
\end{equation*}
where $\eo$ is the amplitude of Eq.~(\ref{eq:2}). From the meromorphic expansion of $\eo$ in Eq.~(\ref{eq:12}) it then follows that,
\begin{equation*}
E(\rr,t)=\re\left\{E_{i}(\rr,t)+ \sum_{n} E_{n}(\rr)\Omega_{n}(t)+\int_{0}^{\infty} A(k)E_{a}(k^2,\rr)e^{-i\omega t}dk\right\}
\label{eq:40}
\end{equation*}
where the time dependence $\Omega_{n}(t)$ of the Siegert state $E_{n}$ is,
\begin{equation}
\Omega_{n}(t)=\widetilde{a}_{n}\sqrt{\Gamma_{n}}\int_{0}^{\infty}\frac{A(k)}{k^2-k_{n}^2+i\Gamma_{n}}e^{-i c k t}dk
\label{eq:41}
\end{equation}
for $\widetilde{a}_{n}$ defined in Eq.(\ref{eq:32}). As $t\rightarrow \infty$, then $\Omega_{n}(t)\rightarrow 0$ as required by the Riemann-Lebesgue lemma. For a typical physical wave packet,
the function $A(k)$ decays fast as $k\rightarrow \infty$
and is also analytic in the complex $k-$plane. These properties
allows for evaluating
the integral (\ref{eq:41}) by the standard means of
the complex analysis. It is also worth noting that
bound states in the radiation continuum cannot be excited
by an incident radiation because $\Gamma_n=0$.

To avoid excessive technicalities of the general case,
a specific form of
the distribution $A(k)$ is chosen to illustrate the procedure,
which is sufficient to establish the main properties
of the decay of Siegert states.
The simplest form that is also most commonly
used in physics is
a Gaussian wave packet centered around a wavenumber $\kc$:
\begin{equation*}
A(k)=\frac{e^{-\frac{(k-\kc)^2}{2\sigma^2}}}{\sigma\sqrt{2\pi}}
\label{eq:42}
\end{equation*}
In the limit $\sigma\rightarrow 0$, $A(k)\rightarrow\delta(k-\kc)$,
and the monochromatic case is recovered.
The analysis will be carried out for a Siegert state $E_{n}$ corresponding to a pole $\es_{n}=k_{n}^{2}-i\Gamma_{n}$ such that $k^{2}_{n}>0$ and $\Gamma_n>0$.
Siegert states with $k^{2}_{n}<0$ are discussed
at the end of this section.

The integrand in (\ref{eq:41}) is highly oscillatory for large $t$,
the contour of integration must be deformed to a curve along which
the phase is constant to obtain a fast converging integral, and thereby to determine the leading term.
To this end, first, the change of variable $\xi=\frac{k-\kc+i c t \sigma^{2}}{\sigma\sqrt{2}}$ is made to obtain the following more amenable form for $\Omega_{n}$:
\begin{equation}
\Omega_{n}(t)=\frac{\widetilde{a}_{n}\sqrt{\Gamma_{n}}e^{-i c t \kc-\frac{1}{2}c^2t^2\sigma^2}}{2\sigma^2\sqrt{\pi}}\int_{\mathcal{C}_{0}} \frac{e^{-\xi^2}}{(\xi-\xi_{+})(\xi-\xi_{-})}\,d\xi,\quad\quad\, \xi_{\pm}=\frac{-\kc+i c t \sigma^2\pm \sqrt{\es_{n}}}{\sigma\sqrt{2}}
\label{eq:43}
\end{equation}
where the contour of integration $\mathcal{C}_{0}$ is a horizontal
ray outgoing from the point $\xi_{0}=\frac{-\kc+i c t \sigma^{2}}{\sigma\sqrt{2}}$ toward the infinity in the $\xi$-plane as shown in Fig.~\ref{fig:2}. The position of the poles $\xi_{\pm}$ indicated on the same figure follows from the fact that for $k^{2}_{n}>0$, then $\re\{\sqrt{\es_{n}}\}>0$ while $\im\{\sqrt{\es_{n}}\}<0$.

Put $(u,v)=(\re\, \xi, \im\, \xi)$.
Since the function $\xi\mapsto e^{-\xi^2}$ decays exponentially in the region $\re\{\xi^2\}>0$, the contour of integration in Eq.(\ref{eq:43}) can be deformed to a contour
that consists of the constant phase curve $\mathcal{C}_{1}=\{(u,v)\,|\, uv =-\frac{\kc c t}{2},\ u\leq \re\,\xi_0\}$ extending from $\xi_{0}$ to $-\infty$ and another constant phase contour $\mathcal{C}_{2}: \im\{\xi\}=0 $ from $-\infty$ to $\infty$. Figure~\ref{fig:2} shows the modified contour. By Cauchy's theorem,
$\Omega_{n}$ becomes the sum of three terms:
\begin{equation*}
\Omega_{n}(t)=\widetilde{a}_{n}(A_{n}(t)+B_{n}(t)+C_{n}(t))
\label{eq:44}
\end{equation*}
where $A_{n}$ is the residual contribution at the pole $\xi_{+}$, while $B_{n}$ and $C_{n}$ are the contributions of the line
integrals along the contours $\mathcal{C}_{1}$ and
$\mathcal{C}_{2}$, respectively. It is proved shortly that it is the residual term $A_{n}$ that accounts for an exponential time decay of the Siegert state $E_{n}$.
For a long period of time, this term
dominates in the decay radiation of a Siegert state, and it is only
after it has decayed considerably that
the term $B_{n}$ becomes dominant. The term $C_{n}$ remains small in comparison to $A_{n}$ and $B_{n}$.

To begin, the residual term $A_{n}$ is evaluated to yield the formula,
\begin{equation*}
A_{n}(t)=\pi i\frac{\sqrt{\Gamma_{n}}}{\sqrt{\es_{n}}}A(\sqrt{\es_{n}}) e^{-i c t\sqrt{\es_{n}}}
\label{eq:45}
\end{equation*}
In particular, the wavenumber $\widetilde{k}_{n}$, and the life time $\tau_{n}$ of the Siegert state given in Eq.(\ref{eq:37}) follow from the formula for $\sqrt{\es_{n}}$ which is,
\begin{equation*}
\sqrt{k_{n}^2-i\Gamma_{n}}=\sqrt{\frac{k_{n}^{2}+\sqrt{k_{n}^{4}+\Gamma_{n}^2}}{2}}-i\frac{\Gamma_{n}}{\sqrt{2\left(k_{n}^{2}+\sqrt{k_{n}^{4}+\Gamma_{n}^2}\right)}}
\label{eq:46}
\end{equation*}
Note that for $\Gamma_n\ll 1$ (a narrow resonance),
the amplitude of the Siegert
state is proportional to  $A(\sqrt{\es_{n}})\approx
A(k_n)$. Hence, the Siegert state corresponding to
the scattering resonance at $k=k_n$ is excited
if the Gaussian wave packet has a sufficient amplitude
at $k=k_n$, or, ideally, is centered at $k_{n}$ to achieve the maximal effect.

The terms $B_{n}$ and $C_{n}$ also decay necessarily in time. For $C_{n}$, the integral along $\mathcal{C}_{2}$ will necessarily remain bounded in time so that the multiplicative factor in Eq.(\ref{eq:43}) is dominant, i.e.,
\begin{equation*}
C_{n}(t)=O(e^{-\frac{1}{2}c^2t^2\sigma^2})
\label{eq:47}
\end{equation*}
The same relation holds for $B_{n}$, provided that
\begin{equation}
ct\leq \frac{\kc}{\sigma^2}
\label{eq:48}
\end{equation}
When this condition is violated, the point $\xi_{0}$ moves into the region $\re\{\xi^2\}<0$, and the exponential $e^{-\xi^2}$ starts to grow large on a part of the contour $\mathcal{C}_{1}$. Nonetheless, the term $B_{n}(t)$ still decays to zero in time. This is because, as noted earlier, the Riemann-Lebesgue lemma ensures
the decay of
 $\Omega_{n}(t)$ in time, and the terms $A_{n}$ and $C_{n}$ have been
proved to converge to zero. For the actual decay rate of $B_{n}$, it can be inferred from Laplace's method for the asymptotic expansion of integrals that $B_{n}=O(1/t^{\alpha})$ for some positive number $\alpha$, which is of no importance here however. Indeed, it is obvious that as $t$ increases, the term $B_{n}$ should be larger than the residue at $\xi_{+}$, since the arc of
the curve $\mathcal{C}_{1}$ in the region $\re\{\xi^2\}<0$ grows longer, and therefore the exponential decay of the Siegert state
can no longer be observed above the background described
by $B_n$. The condition (\ref{eq:48}) shows that
the smaller the width $\sigma$ of the wave packet is, the longer the exponential decay can be observed. For instance,
one can ensure the observation of the decay radiation of
the state $E_{n}$ by requiring that the
half-life time occurs before the decay amplitude becomes
smaller than that of the background, i.e.,
\begin{equation*}
c\tau_{n}\leq \frac{\kc}{\sigma^{2}}\quad
\Longleftrightarrow\quad \frac{\ln 2\sqrt{k_{n}^{2}+\sqrt{k_{n}^{4}+\Gamma_{n}^{2}}}}{\kc}\leq \frac{\Gamma_{n}}{\sigma^2}
\label{eq:49}
\end{equation*}
In other words, the width $\sigma$ of the packet must be comparable to $\sqrt{\Gamma_{n}}$, if not smaller.
\begin{figure}[h t]
    \centering
    \includegraphics[scale=0.75]{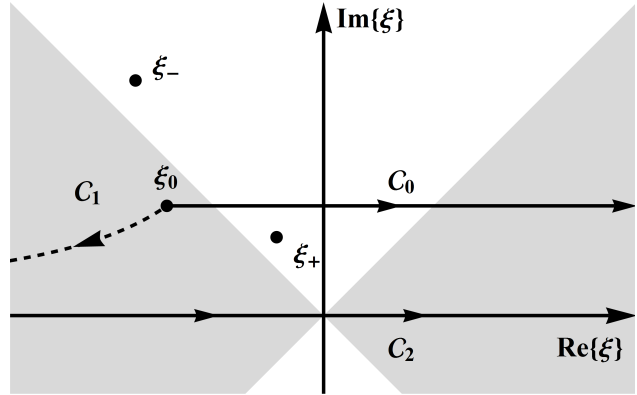}
    \caption{The $\xi$-plane for the integral in Eq.~(\ref{eq:43}). The region $\re\{\xi^2\}\geq 0$ is shadowed.
The exponential $e^{-\xi^2}$ is bounded in it. The sketch is realized under the condition (\ref{eq:48}), when the point $\xi_{0}$ is still in the shadowed region. When the condition is violated, part of the curve $\mathcal{C}_{1}$ lies in the region $\re\{\xi^2\}< 0$.}
    \label{fig:2}
    \end{figure}

For Siegert states with $k^{2}_{n}<0$, the pole $\xi_{+}$ is no longer enclosed by the contours $\mathcal{C}_{0},\,\mathcal{C}_{1}$, and $\mathcal{C}_{2}$ if $t$ is sufficiently large. The pole moves to the left of the contour $\mathcal{C}_{1}$ and the exponential time decay is absent, owing to the fact that the incident wave packet does not contain any radiation at wavenumbers close to $k_{n}$ for the Siegert state to be excited.

\appendix

\begin{center}
{\bf \large Appendix: Proof of the Regular Perturbation Theorem}
\end{center}

Recall that $\forall (h,\es)\in J_{0}\times \ckx, \, \hv(h,\es)\in \mathcal{L}\{L^{2}(D)\}$, and that the map
\begin{equation*}
\begin{split}
J_{0}\times\ckx&\rightarrow \mathcal{L}\{L^{2}(D)\}\\
(h,\es)&\mapsto \hv(h,\es)
\end{split}
\label{eq:a1}
\end{equation*}
is continuously Frechet differentiable. For $(h,\es)\in J_{0}\times\ckx$, let $R_{\lambda}(\hv(h,\es))=[\lambda-\hv(h,\es)]^{-1},\, \lambda\in \mathbbm{C},$ be the resolvent of $\hv(h,\es)$ when it exists, and let $\rho(\hv(h,\es))=\{\lambda\in\mathbbm{C}| R_{\lambda}(\hv(h,\es)) \,\,\text{exists}\}$ be
the resolvent set of $\hv(h,\es)$.
The proof of the Regular Perturbation Theorem is aided
with the following lemma:

\begin{lemma}
Let $\se=\{(h,\es,\lambda)\in J_{0}\times\ckx\times
\mathbbm{C}| \lambda\in \rho(\hv(h,\es))\}$. The set $\se$ is open in  $J_{0}\times\ckx\times\mathbbm{C}$ and nonempty.
\end{lemma}

\begin{proof}
That $\se$ is nonempty follows from the results
of Section~\ref{sec:2}. Indeed, if $\es$ is not a pole of $\es\mapsto [1-\hv(h,\es)]^{-1}$, then $(h,\es,1)\in \se$. To show that $\se$ is open, let $(h_{0},\es_{0},\lambda_{0})\in \se$. Then one has
to prove that there always exists
a neighborhood of $(h_{0},\es_{0},\lambda_{0})$ in $J_{0}\times\ckx\times \mathbbm{C}$ that lies in $\se$.

Since $\rho(\hv(h_{0},\es_{0}))$ is open in $\mathbbm{C}$, there exists $\lambda_{1}\in \rho(\hv(h_{0},\es_{0})),\, \lambda_{1}\neq \lambda_{0}$ so that $[\lambda_{1}-\hv(h_{0},\es_{0})]^{-1}$ exists. Let then $\Psi$ be the function,
 \begin{equation*}
 \begin{split}
 \Psi : \, J_{0}\times\ckx &\rightarrow \mathcal{L}\{L^{2}(D)\}\\
   (h,\es) &\mapsto \lambda_{1}-\hv(h,\es)
   \end{split}
   \label{eq:a2}
   \end{equation*}
The map $\Psi$ is continuous, and $\Psi(h_{0},\es_{0})$ is invertible. As the set of all invertible operators in $\mathcal{L}\{L^{2}(D)\}$ is open, there exists an open neighborhood $U\subset J_{0}\times\ckx$ of $(h_{0},\es_{0})$ such that $\forall (h,\es)\in U$, $\Psi(h,\es)$ is invertible.

Now by the first resolvent formula,
\begin{equation*}
R_{\lambda_{0}}(\hv(h_{0},\es_{0}))-R_{\lambda_{1}}(\hv(h_{0},\es_{0}))=(\lambda_{1}-\lambda_{0})R_{\lambda_{0}}(\hv(h_{0},\es_{0}))R_{\lambda_{1}}(\hv(h_{0},\es_{0}))
\label{eq:a3}
\end{equation*}
so that,
\begin{equation*}
R_{\lambda_{0}}(\hv(h_{0},\es_{0}))\left(1-(\lambda_{1}-\lambda_{0})R_{\lambda_{1}}(\hv(h_{0},\es_{0}))\right)=R_{\lambda_{1}}(\hv(h_{0},\es_{0}))
\label{eq:a4}
\end{equation*}
It follows that $1-(\lambda_{1}-\lambda_{0})R_{\lambda_{1}}(\hv(h_{0},\es_{0}))$ is invertible. Now let $\Phi$ be the function,
\begin{equation*}
\begin{split}
\Phi:\,U\times\mathbbm{C}&\rightarrow \mathcal{L}\{L^{2}(D)\}\\
\left((h,\es),\lambda\right)&\mapsto 1-(\lambda_{1}-\lambda)R_{\lambda_{1}}(\hv(h,\es))
\end{split}
\label{eq:a5}
\end{equation*}
Then $\Phi$ is continuous. Since $\Phi(h_{0},\es_{0},\lambda_{0})$ is invertible, and again, the set of all invertible operators in $\mathcal{L}\{L^{2}(D)\}$ is open, it follows that there exists an open set $W\subset U\times\mathbbm{C}$ such that for all $(h,\es,\lambda)\in W$, the operator $\Phi(h,\es,\lambda)$ is invertible.

Thus, for $(h,\es,\lambda)\in W$, both $\lambda_{1}-\hv(h,\es)$ and $1-(\lambda_{1}-\lambda)R_{\lambda_{1}}(\hv(h,\es))$ are invertible. Their composition is therefore invertible, and this composition is,
\begin{equation*}
(\lambda_{1}-\hv(h,\es))(1-(\lambda_{1}-\lambda)R_{\lambda_{1}}(\hv(h,\es)))=\lambda-\hv(h,\es)
\label{eq:a6}
\end{equation*}
It follows that $W\subset \se$. As $W$ is a neighborhood of $(h_{0},\es_{0},\lambda_{0})$ in $J_{0}\times\ckx \times \mathbbm{C}$,
 the set $\se$ is open.
\end{proof}

The proof of the Regular Perturbation Theorem is as follows.
\begin{proof}
Let $\es_{n}$ be a simple pole of $\es\mapsto [1-\hv(h_{n},\es)]^{-1}$ for some fixed $h_{n}\in J_{0}$, and let $E_{n}$ be the corresponding Siegert state. Then 1 is an eigenvalue of $\hv(h_{n},\es_{n})$ corresponding to the eigenfunction $E_{n}$ in the usual sense. The theorem amounts to proving the existence of a curve $h\mapsto \es(h)$ in $J_{0}\times\ckx$ along which the family of operators $\hv(h,\es(h))$ still has 1 as an eigenvalue. The corresponding eigenfunctions will then be the Siegert states of the said operators along the curve in question.

The proof starts by providing a general formula for the eigenvalue $\lambda_{0}(h,\es)$ of $\hv(h,\es)$ which is the perturbed value of the eigenvalue 1 when the point $(h_{n},\es_{n})$ is displaced to $(h,\es)$ in the $J_{0}\times \ckx$ space. Then by the Implicit Function Theorem, a curve $h\mapsto\es(h)$ is found along which the eigenvalue $\lambda(h,\es)$ remains 1, i.e., $\lambda(h,\es(h))=1$.

As a starting point, note that as $\hv(h_{n},\es_{n})$ is a compact operator on the Hilbert space $L^{2}(D)$, the Riesz-Schauder
theorem \cite{b2} implies that 1 is necessarily an isolated eigenvalue. Therefore, there exists $\delta>0$ such that 1 is the only eigenvalue of $\hv(h_{n},\es_{n})$ in the disk $\{\lambda
\in \mathbbm{C}\ |\ |\lambda-1|\leq \delta\}$. It follows that $V_{\delta}=\{(h_{n},\es_{n},\lambda)|\lambda\in \mathbbm{C},\, |\lambda-1|=\delta\}\subset \mathcal{S}\}$ where $\mathcal{S}$ is the set of the previous lemma. Since $\mathcal{S}$ is open, and $V_{\delta}$ is compact, there exists an open $W$ in $J_{0}\times\ckx
\times\mathbbm{C}$ such that $V_{\delta}\subset W\subset \mathcal{S}$.  Therefore there exists a connected neighborhood $U$ of $(h_{n},\es_{n})$ in $J_{0}\times \ckx$ such that $\lambda\in \rho(\hv(h,\es))$
for all $(h,\es)\in U$, and for all $\lambda\in\mathbbm{C}$ with $|\lambda-1|=\delta$.

Put
\begin{equation}
\pv(h,\es)=\frac{1}{2\pi i}\oint_{|\lambda-1|=\delta} [\lambda-\hv(h,\es)]^{-1} d\lambda,\quad \, (h,\es)\in U
\label{eq:a7}
\end{equation}
Then $\pv(h,\es)\in\mathcal{L}\{L^{2}(D)\}$, and by Theorem XII.6 of~\cite{b3}, $\pv(h,\es)$ is a projection. As the map $(h,\es)\mapsto \pv(h,\es)$ of $U$ to $\mathcal{L}\{L^{2}(D)\}$ is continuous, it follows that the dimension of the range of $\pv(h,\es)$ is constant throughout $U$.

Since the eigenvalue 1 of $\hv(h_{n},\es_{n})$ is nondegenerate, it follows that the range of $\pv(h_{n},\es_{n})$ has dimension 1, and therefore the dimension of the range of $\pv(h,\es))$ is 1 throughout $U$. By Theorem XII.6 of~\cite{b3}, it follows that for all $(h,\es)\in U$, there exists a unique eigenvalue $\lambda_{0}(h,\es)$ of $\hv(h,\es)$ in $\{\lambda\in\mathbbm{C}\,|\, |\lambda-1|\leq \delta\}$, and $\pv(h,\es)$ is the projection on the corresponding eigenspace. In particular, $\pv(h,\es)[E_{n}]$ is in the eigenspace of $\lambda_{0}(h,\es)$, and therefore,
\begin{equation*}
\hv(h,\es)\left[\pv(h,\es)[E_{n}]\right]=\lambda_{0}(h,\es) \pv(h,\es)[E_{n}],\,\quad (h,\es)\in U
\label{eq:a8}
\end{equation*}
Hence,
\begin{equation*}
\lambda_{0}(h,\es)=\frac{\langle E_{n}, \hv(h,\es)\left[\pv(h,\es)[E_{n}]\right]\rangle}{\langle E_{n},\pv(h,\es)[E_{n}]\rangle},\quad
\, (h,\es)\in U
\label{eq:a10}
\end{equation*}
where once again, $\langle \cdot, \cdot \rangle$ is the inner product on $L^{2}(D)$. This is the formula for the perturbed eigenvalue announced in the preamble of this proof.

Now, observe that since $(h,\es)\mapsto \hv(h,\es)$ is continuously Frechet differentiable in $U$, so is $(h,\es)\mapsto \pv(h,\es)$. Therefore, $\lambda_{0}$ is continuously Frechet differentiable in $U$. Now, $\lambda_{0}(h_{n},\es_{n})=1$, and as it will be shown shortly,  $\partial_{\es}\lambda_{0}(h_{n},\es_{n})\neq 0$. It follows by the Implicit Function Theorem for Banach spaces that there exists an open interval $J\subset J_{0}$ containing $h_{n}$, and a unique continuously differentiable function $h\mapsto \es_{n}(h),\, h\in J,$ such that $\es_{n}(h_{n})=\es_{n}$ and $\lambda_{0}(h,\es_{n}(h))=1$ for all $h\in J$. Thus the proof of the Regular perturbation theorem is complete, provided the
property $\partial_{\es}\lambda_{0}(h_{n},\es_{n})\neq 0$
is established. The latter is achieved by studying the analyticity of the resolvent $R_{\lambda}(\hv(h_{n},\es))=[\lambda-\hv(h_{n},\es)]^{-1}$ in the variables $\lambda$ and $\es$ separately
 when $h=h_{n}$ is fixed.

Let $\lambda\in\mathbbm{C}$ and $\es\in\ckx$ such that $|\lambda-1|=\delta$, $(h_{n},\es)\in U$, and $\es\neq\es_{n}$. Then the resolvent $R_{\lambda}(\hv(h_{n},\es))=[\lambda-\hv(h_{n},\es)]^{-1}$ exists by definition of the neighborhood $U$. Also, the resolvent $ R_{1}(\hv(h_{n},\es))=[1-\hv(h_{n},\es)]^{-1}$ exists and is meromorphic in $\es$. By the first resolvent formula,
\begin{equation}
R_{\lambda}(\hv(h_{n},\es))-R_{1}(\hv(h_{n},\es_{n}))=(1-\lambda)R_{\lambda}(\hv(h_{n},\es))R_{1}(\hv(h_{n},\es))
\label{eq:a11}
\end{equation}
Substituting Eq.(\ref{eq:a11}) into Eq.(\ref{eq:a7}), it follows that,
\begin{equation}
\pv(h_{n},\es)=\frac{1}{2\pi i}\oint_{|\lambda-1|=\delta}(1-\lambda)R_{\lambda}(\hv(h_{n},\es))R_{1}(\hv(h_{n},\es))d\lambda
\label{eq:a12}
\end{equation}
Now, the meromorphic expansion of $\lambda\mapsto R_{\lambda}(\hv(h_{n},\es))$ at $\lambda_{0}(h_{n},\es)$ is,
\begin{equation}
R_{\lambda}(\hv(h_{n},\es))=\frac{1}{\lambda-\lambda_{0}(h_{n},\es)} \pv(h_{n},\es)+\widetilde{R}_{\lambda}(h_{n},\es)
\label{eq:a13}
\end{equation}
where $\lambda\mapsto\widetilde{R}_{\lambda}(h_{n},\es)$ is analytic in the disk $\{\lambda\in\mathbbm{C}:\, |\lambda-1|\leq \delta\}$. The substitution of Eq.(\ref{eq:a13}) into Eq.(\ref{eq:a12}) yields
\begin{equation}
\pv(h_{n},\es)=(1-\lambda_{0}(h_{n},\es))\pv(h_{n},\es)R_{1}(\hv(h_{n},\es))
\label{eq:a14}
\end{equation}
Now, at the pole $\es_{n}$, the generalized resolvent $\es\mapsto R_{1}(h_{n},\es)$ has the meromorphic expansion:
\begin{equation}
R_{1}(h_{n},\es)=\frac{1}{\es-\es_{n}}\hv_{n}+\widetilde{R}(h_{n},\es)
\label{eq:a15}
\end{equation}
where $\hv_{n}$ is the residue of $\es\mapsto R_{1}(h_{n},\es)$ as given in Eq.(\ref{eq:11}), and $\es\mapsto \widetilde{R}(h_{n},\es)$ is analytic in the vicinity of $\es_{n}$. Substituting Eq.(\ref{eq:a15}) into Eq.(\ref{eq:a14}) and taking the limit as $\es\rightarrow \es_{n}$ yields the equation,
\begin{equation}
\pv(h_{n},\es_{n})=-\partial_{\es}\lambda_{0}(h_{n},\es_{n})\pv(h_{n},\es_{n})\hv_{n}
\label{eq:a16}
\end{equation}
Next, the action of both the sides of this operator equality is
evaluated on the state $E_{n}$. Since $\pv(h_{n},\es_{n})$ projects on $E_{n}$, and in light of Eq.(\ref{eq:11}), it follows that,
\begin{equation*}
E_{n}=-\langle \varphi_{n}, E_{n}\rangle \partial_{\es}\lambda_{0}(h_{n},\es_{n})E_{n}
\label{eq:a17}
\end{equation*}
As the Siegert state $E_{n}$ is not identically zero, it follows that $\partial_{\es}\lambda_{0}(h_{n},\es_{n})\neq 0$, and the proof of the theorem is complete.

A final noteworthy remark is that the projection $\pv(h_{n},\es_{n})$ and the residue $\hv_{n}$ are proportional,
which is established by applying both the sides of
Eq.(\ref{eq:a16}) to an arbitrary state $F\in L^{2}(D)$:
\begin{equation*}
\pv(h_{n},\es_{n})[F]=-\partial_{\es}\lambda_{0}(h_{n},\es_{n})\hv_{n}[F]
\label{eq:a18}
\end{equation*}
\end{proof}

\end{document}